\documentclass[aps, prl, a4paper,showpacs,twocolumn,10pt]{revtex4-1}
\usepackage{bbm, amsmath, amssymb, amsthm, bm,textcomp, nicefrac,geometry,ragged2e}
\usepackage{graphicx,epstopdf,color,verbatim,enumitem}
\usepackage[dvipsnames]{xcolor}
\usepackage{float}
\usepackage[bbgreekl]{mathbbol}
\usepackage{graphicx,epstopdf,color,verbatim,enumitem,ulem}
\usepackage{mathtools}

\definecolor{myurlcolor}{rgb}{0,0,0.7}
\definecolor{myrefcolor}{rgb}{0.8,0,0}
\usepackage[unicode=true,pdfusetitle, bookmarks=false,bookmarksnumbered=false,
bookmarksopen=false, breaklinks=false,pdfborder={0 0 0},backref=false,
colorlinks=true, linkcolor=myrefcolor,citecolor=myurlcolor,urlcolor=myurlcolor]
{hyperref}

\usepackage{pbox,hyperref,array}
\usepackage[caption=false]{subfig}

\newcommand{\ket}[1]{\left|#1\right\rangle}

\newcommand{\braket}[2]{\langle #1 | #2 \rangle }

\newcommand\cH{{\mathcal H}}

\newcommand{\be}{\begin{equation}}
\newcommand{\ee}{\end{equation}}

\definecolor{darkorange}{RGB}{255,140,0}

\newcommand{\bea}{\begin{eqnarray}}
\newcommand{\eea}{\end{eqnarray}}

\def\mn#1{\langle #1 \rangle}

\makeatletter
\newtheorem*{rep@theorem}{\rep@title}
\newcommand{\newreptheorem}[2]{%
\newenvironment{rep#1}[1]{%
 \def\rep@title{#2 \ref{##1}}%
 \begin{rep@theorem}}%
 {\end{rep@theorem}}}
\makeatother

\newreptheorem{theorem}{Theorem}

\newtheorem{proposition}{Proposition}

\newtheorem*{result*}{Result}

\DeclareGraphicsExtensions{.png,.pdf,.eps}

\begin{document}
\title{Device-independent detection of genuine multipartite entanglement for all pure states}
\author{M. Zwerger$^1$, W. D\"ur$^1$, J.-D. Bancal$^2$ and P. Sekatski$^2$}
\affiliation{$^1$ Institut f\"ur Theoretische Physik, Universit\"at Innsbruck, Technikerstra{\ss}e 21a, 6020 Innsbruck, Austria\\
$^2$ Departement Physik, Universit\"at Basel, Klingelbergstra{\ss}e 82, 4056 Basel, Switzerland}

\date{\today}

\begin{abstract}
We show that genuine multipartite entanglement of all multipartite pure states in arbitrary finite dimension can be detected in a device-independent way by employing bipartite Bell inequalities on states that are deterministically generated from the initial state via local operations. This leads to an efficient scheme for large classes of multipartite states that are relevant in quantum computation or condensed-matter physics, including cluster states and the ground state of the Affleck-Kennedy-Lieb-Tasaki (AKLT) model. 
For cluster states the detection of genuine multipartite entanglement involves only measurements on a constant number of systems with an overhead that scales linear with the system size, while for the AKLT model the overhead is polynomial. In all cases our approach shows robustness against experimental imperfections.
\end{abstract}
\maketitle

\paragraph{Introduction.---}
Entanglement is an exclusive feature of quantum physics. As such it is believed to be the key ingredient in various quantum information processing tasks, like e.g. quantum computation, quantum metrology and, to some extent, quantum key distribution. Entanglement is a direct consequence of the fact that quantum states are modeled as operators on the tensor product of the Hilbert spaces for each system. From this mathematical perspective the question of entanglement detection has been solved, most notably by the approach based on entanglement witnesses \cite{Guhne2009}. 

However, from a more physical perspective such a view is not fully satisfactory, as in order to be applied to an experiment it requires to assume a given dimension for the Hilbert space of each system and an exact quantum description of the measurement devices. Yet, it is hard to characterize a measurement device exactly, and moreover a physical system typically has access to more levels and degrees of freedom than one uses to describe its state. Hence neither of the assumptions can be fully verified in practice. 

The most radical way to overcome these problems is offered by device-independent methods, which allow one to detect entanglement solely based on the Bell-like correlations of measurement outcomes collected in the experiment. While for the bipartite case many results of fundamental interest have been obtained \cite{Gisin1991,BellRev,Scarani2017,DIbipartite,DIbipartite2}, less is known for multipartite case. In particular, when it comes to genuine device-independent entanglement, results are only known for a few states \cite{Svetlichny1987,Collins2002,Seevinck2002,Bancal2010,Bancal2011,Bancal2011b,Chen2011,Curchod2014,McKague2016,Baccari2018}. This has to do with the difficulty to obtain multipartite Bell inequalities suited to specific states: typically, the starting point of the analysis is a fixed set of known Bell inequalities rather than the states themselves. In addition, multipartite Bell inequalities such as the Svetlichny inequality are inefficient to test experimentally, as they require an exponentially increasing number of measurement settings.

Here we circumvent these difficulties by introducing a scheme that allows one to detect genuine multipartite entanglement by testing bipartite Bell inequalities on states that are generated deterministically from the initial state via local operations and classical communication(LOCC) (see also \cite{Almeida2010,Supic2017} for a related approach). By using a covering set of such pairs, we derive a multipartite Bell inequality and show that a sufficiently large violation of the bipartite inequalities allows one to certify genuine multipartite entanglement in a device-independent way. This approach is not restricted to pure states, but has also some built-in robustness against noise and imperfections. We show genuine multipartite entanglement for all entangled pure states with arbitrary (finite) local dimension, and also for all mixed states that are sufficiently close to any such state. What is more, we obtain a scheme that is experimentally efficient for large classes of interesting states, including all (weighted) graph states with constant degree \cite{Raussendorf2003,He04,He06,Anders2006} and ground states of 1d spin models such as the AKLT model \cite{AKLT}. That is, only a constant (logarithmically growing) number of parties needs to be measured, and the overhead in terms of measurement settings is only linear (polynomial) in the system size $N$ respectively, as opposed to previous schemes that scale exponentially with $N$.

\paragraph{Statement of the main results.---}
Let $N \in \mathbb{N}$ denote the number of parties, $V=\{1, ...,N\}$, $E_k = \{i_k, j_k\}$ with $i_k, j_k\in V$ be pairs of parties and $E= \{E_1, ..., E_K\}$ their union. We define a graph $G= (V, E)$ by associating the parties with vertices and the pairs $E_k$ with edges. We say that $E$ is a \textit{covering set of pairs} for the $N$-partite system if the corresponding graph $G$ is connected.

The main result of this manuscript is the following:

\textit{Theorem 1}

Let $\ket{\psi}$ be a state in the Hilbert space $\otimes_{i=1}^{N}\mathbb{C}^{d_i}$, where $d_i \in \mathbb{N}$ for all $i \in \{1, ..., N\}$. If there exists a covering set such that for each pair in it there exist local operations and classical communication such that one can produce an entangled pure state between the two parties in all branches of the LOCC protocol, then one can show genuine multipartite entanglement between all $N$ parties in a device-independent way.

\textit{Sketch of the proof:}
The idea behind the proof is the following. One constructs a Bell expression by considering the sum of all bipartite Bell expressions for the pairs of parties appearing in the covering set $E$ and all branches of the LOCC protocol (that all result in a pure entangled state of the chosen pair $E_k$). The bipartite Bell inequalities are chosen to reach the quantum bound $\beta_{*}$ for each of the states, which is in fact possible due to a recent result \cite{Scarani2017}. One then shows that the quantum state achieves a value for the multipartite Bell expression which is incompatible with any biseparable quantum state
\be \label{eq: state exp}
\varrho_\textrm{\,BS}= \sum_\lambda\!\!\!\! \sum_{\substack{ \\ g_1\cap \,g_2 =\emptyset \\g_1\cup\, g_2 =\{1, \dots,  N \}}}\!\!\!\!
p\,(\lambda)\,  \rho_{g_1}(\lambda) \otimes \rho_{g_2}(\lambda)
\ee
of arbitrary local dimensions. Here the sum over $g_1$ and $g_2$ cover all possible splittings of the $N$ parties in two groups, $\rho_{g_1}$ are $\rho_{g_2}$ are arbitrary joint quantum states of the parties belonging to the corresponding group, and  $\lambda$ is the hidden variable.

For each pair of parties and each component of the biseparable state expansion appearing in Eq.~\eqref{eq: state exp} there are two possibilities regarding the expectation value of the bipartite Bell expressions in the multipartite Bell inequality. Either the two parties appearing in the bipartite inequality belong to same group $g_1$ or $g_2$, in which case they may always achieve the maximal value of the Bell expression $\beta_{*}$. Or they belong to different groups $g_1$ and $g_2$, in which case they end up in a separable state in each branch of the LOCC protocol and can at most contribute a value corresponding the local bound $\beta_{L} < \beta_{*}$. Since $E$ is a covering set of pairs, for each grouping $g_1|g_2$ in the biseparable state expansion of Eq.~\eqref{eq: state exp} there will be at least one term in the Bell expression, corresponding to some pair $E_k$, whose value is limited to $\beta_L<\beta^*$. Hence, the expected value of the overall Bell expression on a biseparable state $\rho_\textrm{\,BS }$ is also strictly smaller than $\beta^*$. For a detailed proof see the appendix.

In fact, in the ideal case, where the observed violation of the constructed Bell expression is maximal $\beta^*$, we show that it is incompatible with any quantum state $(1-\varepsilon_\text{BS})\varrho_Q + \varepsilon_\text{BS}\,\varrho_\textrm{\,BS}$ that has a non-zero biseparable weight $\varepsilon_\text{BS}$, i.e., with any mixture of a genuinely multipartite entangled state and a biseparable state with arbitrarily small nonzero weight. We will also use this property to show the robustness our result later in the paper.

\textit{Theorem 2}

All genuine entangled pure states fulfill the requirements from Theorem 1.

\textit{Sketch of the proof:}
One can show that almost any measurement brings a genuinely entangled $N$-partite state $\ket{\Psi}$ to $d_1$ post-measurement $(N-1)$-partite genuinely entangled states $\ket{\Psi_k}$, where $d_1$ is the local Hilbert space dimension. This can be iterated until one arrives at a bipartite state. The set of measurements for which it does not work is of measure zero in each step. This guarantees that there are measurements on $N-2$ parties such that one obtains an entangled bipartite state in all branches. For a detailed proof see the appendix.

While this works for all pure states, in general the protocol is not efficient. From the proof of Theorem 2 one sees that up to $N-2$ parties need to carry out measurements with at least two outcomes each. This gives rise to exponentially many branches in which one needs to test a bipartite Bell inequality. However, there are families of states for which only a limited number of parties are involved in the LOCC protocol for each pair, and thus the protocol is efficient. This is described in more detail in the next section.

The results above are very general. We now illustrate its usefulness with some examples of classes of states for which one can (efficiently) show genuine multipartite entanglement in a device-independent way.

\paragraph{Connected, generalized and weighted graph states.---}
This family of states plays an important role in quantum information, in particular in the context of quantum error correction, measurement-based quantum computation and quantum networks. It has first been defined for qubits \cite{Raussendorf2003,He04,He06} and later been generalized to higher local dimensions \cite{Zhou2003,Bahramgiri2006}. The toric code and its generalizations \cite{Kitaev2006} also belong to this family. For qubits, weighted graph states \cite{Anders2006} can be defined by $\ket{G} = \prod_{\{i,j\} \in E}U_{ij} \ket{+}^{\otimes N}$ Here, $U_{ij}= \rm{diag}(1, 1, 1, e^{i\phi_{ij}})$ in the computational basis, $\ket{+}$ is the $+1$ eigenstate of the Pauli $X$ operator and $V$ and $E$ are sets of vertices and edges as above. If $G=(V,E)$ is connected, then the (weighted) graph state $\ket{G}$ is said to be connected. For a discussion of the case of local dimension $d > 2$ see \cite{Zhou2003,Bahramgiri2006}.

It is easy to see that these states fulfill the requirements from Theorem 1. The set $E$ itself is a covering set of pairs and any pair in it can be isolated via measurements of all qubits in the neighborhood in the computational basis \cite{Raussendorf2003,He04,He06}, since the gate $U_{ij}$ commutes with the measurement in the $Z$ basis. The efficiency of the protocol depends on the number of outcomes for the measurements of the qubits in the neighborhood of each pair. The number of neighbors is specified by the degree $\mathrm{deg}(G)$ of the graph, and is at most $2\cdot\mathrm{deg}(G)$. For all measurement outcomes, one obtains a state that is equivalent up to local Pauli corrections to $U_{ij}\ket{+}^{\otimes 2}$. The covering set can always be chosen to contain at most $N$ pairs. This can be achieved by first choosing one vertex and adding all edges connecting this vertex to the set $E'$. One then continues this step for all neighbors of this vertex, but adds only those edges that connect vertices which were not already connected in the previous round. The size of the neighborhood, i.e., the number of vertices adjacent to a pair, enters exponentially in the total number of operators which need to be measured. This is because one has to take all measurement outcomes into account. For qubit graph states one has to optimize over all possible covering sets and over all local unitary (LU) equivalent states for each pair individually. The concept of local complementation can substantially change the degree of a graph, e.g. for a binary tree graph a sequence of local complementations \cite{He04,He06} can change the degree from $3$ to $N-1$ and vice versa. 
In particular, as long as the maximal degree of the graph grows at most logarithmically with number of vertices $N$ the protocol is efficient. For constant degree one indeed obtains a linear scaling, as there are only linearly many terms in the Bell inequality, and each has support on a constant number of parties only. 
This holds e.g. for prominent graphs states defined on square and triangular lattices, which are also universal resources for measurement-based quantum computation \cite{Ra2001,He06}.

\paragraph{Affleck-Kennedy-Lieb-Tasaki model.---}
The AKLT model \cite{AKLT} is a generalization of the one-dimensional (quantum) Heisenberg spin model, with Hamiltonian $H=\sum_j {\bf S}_{(j)}\cdot {\bf S}_{(j+1)} + \tfrac{1}{3}({\bf S}_{(j)}\cdot {\bf S}_{(j+1)})^2$ where ${\bf S}_j$ is the spin-1 operator acting on system $j$. The model is exactly solvable and can be viewed as a prototype of a matrix product state (MPS) \cite{Fannes1992} (for reviews see \cite{Verstraete2008,Orus2014}). One can certify genuine multipartite entanglement in the AKLT model in a device-independent way efficiently. The preparation of entangled states of pairs $(j,j+1)$, as required in Theorem 1, can be achieved by measuring only a small neighborhood of each pair, and for all measurement outcomes one is left with an entangled pair. Using the notation introduced in Fig. \ref{AKLT}, it suffices to measure the neighboring spins of the pair in the basis $\{ \ket{\tilde0}, \ket{\tilde1}, \ket{\tilde2} \}$, where each outcome occurs with probability one third. For outcomes corresponding to $\ket{\tilde0}$ and $\ket{\tilde1}$ the chain is decoupled, and for the outcome corresponding to $\ket{\tilde2}$, which corresponds to entanglement swapping at the level of the virtual links, one has shifted the problem of cutting to the next site. One can then repeat the probabilistic cutting. Measuring $n$ sites on each side of the pair results in a success probability of cutting out the pair of $p_{\textrm{cut}} \ge 1- 2(\tfrac{1}{3})^n$. The state of the resulting pair depends on the outcome of the measurements, and is always entangled and pure \footnote{It is maximally entangled only for some of the measurement outcomes.}.
The success probability goes to one exponentially fast with $n$. From the results presented below, it follows that $n=O(\log N)$, as such a reduced success probability yields a smaller, non-maximal violation of the bipartite inequality. This corresponds to a polynomial number of measurement settings and hence an efficient scheme to detect genuine entanglement in the AKLT model. See also the appendix.

\begin{figure}[ht]
\centering
\includegraphics[scale=0.5]{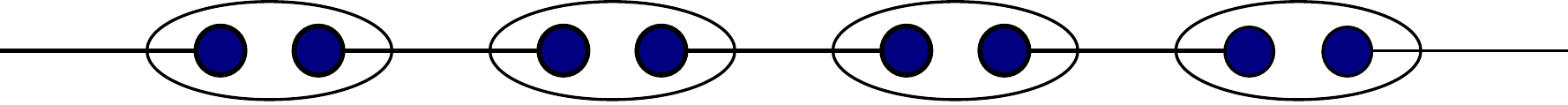}
\vspace{0mm}
\caption{Pictorial representation of the AKLT state. There are two virtual qubits (blue dots) at each site. Dots connected by an edge represent singlet states $\ket{\psi^{-}}$ and ellipses refer to projections onto the three-dimensional triplet subspace, where one makes the following identification: $\ket{\tilde0}=\ket{00}, \ket{\tilde1}=\ket{11},\ket{\tilde2}=\ket{\psi^+} = \tfrac{1}{\sqrt{2}} \left( \ket{01} + \ket{10} \right)$.}
\label{AKLT}
\end{figure}

\paragraph{Dicke states.---}
Dicke states \cite{Dicke} are an important class of multipartite entangled states. An $N$-qubit Dicke state with $k$ excitations is given by $\ket{D_k^N} = {N \choose k}^{-1/2} \sum_{\mathrm{permutations}} \ket{1}^{\otimes k} \ket{0}^{\otimes N-k}$, where the sum refers to all permutations of the parties. Entangled states for any pair of parties can be produced deterministically via Pauli $Z$ and $X$ measurements (see the appendix) and thus it follows from Theorem 1 that one can show genuine multipartite entanglement in a device-independent way.

\paragraph{Robustness and experimental feasibility.---}
The robustness to imperfections is crucial for the experimental feasibility of any protocol. We show that our method to reveal genuine multipartite entanglement is robust to noise and study two different situations.
From the analysis presented in the appendix it follows that in the presence of a non-maximal violation of the bipartite inequality one can show genuine entanglement only for up to
\be
M=  \left \lfloor \frac{\beta_*-\beta_L}{\beta_*-\beta} \right \rfloor 
\ee
parties, where $\beta$ is the observed value of the bipartite inequality. One sees that the robustness is determined by the ratio of the local bound $\beta_{L}$ and the quantum bound $\beta_{*}$. However, one obtains a certain robustness for \textit{any} entangled pure state. That is, there exists an $\epsilon$-ball of non-zero measure around each entangled pure state where we can confirm genuine entanglement with our method. Consider $\rho=(1-\epsilon)\rho +\epsilon/d^N I$, i.e., a mixture of the state with the identity on the whole space of $N$ qudits. In this case the observed value $\beta=(1-\epsilon)\beta_*$, and we obtain that we have genuine $N$-party entanglement if $\epsilon \leq (1-\beta_L/\beta_*))/N$. When using the tilted CHSH Bell inequality for qubits \cite{tCHSH}, the robustness is connected to the amount of entanglement of the produced pairs, as the ratio of $\beta_L/\beta_*$ is smaller for pairs with more entanglement, leading to a better robustness. Finding new Bell inequalities with a better ratio will therefore improve the experimental feasibility.

We now turn to an explicit example and consider the impact of local depolarizing noise (LDN) acting on each qubit of a cluster state.  LDN can be viewed as a worst case local noise model \cite{Dur2005}. It is parametrized by $p \in [0, 1]$, where $p=1$ corresponds to no noise and $p=0$ to complete depolarization and is described by a map ${\cal E}(p)\rho = p\rho + (1-p)/4 \sum_j\sigma_j \rho \sigma_j$. We choose $1D$ and $2D$ cluster states for testing the robustness and assume an infinite system size (or periodic boundary conditions). In Fig. \ref{Figure_CHSH} we plot the number of parties $M$ for which one can show genuine multipartite entanglement as a function of the noise $1-p$. We choose a covering set which only contains nearest neighbor pairs. For cluster states it is possible to establish maximally entangled states between any pairs and hence we employ the CHSH inequality \cite{CHSH} as the bipartite Bell inequality, which has $\beta_{*}=2\sqrt{2}$ and $\beta_{L}=2$.

\begin{figure}
 \centering
 \subfloat[\centering]{\includegraphics[width=1\columnwidth]{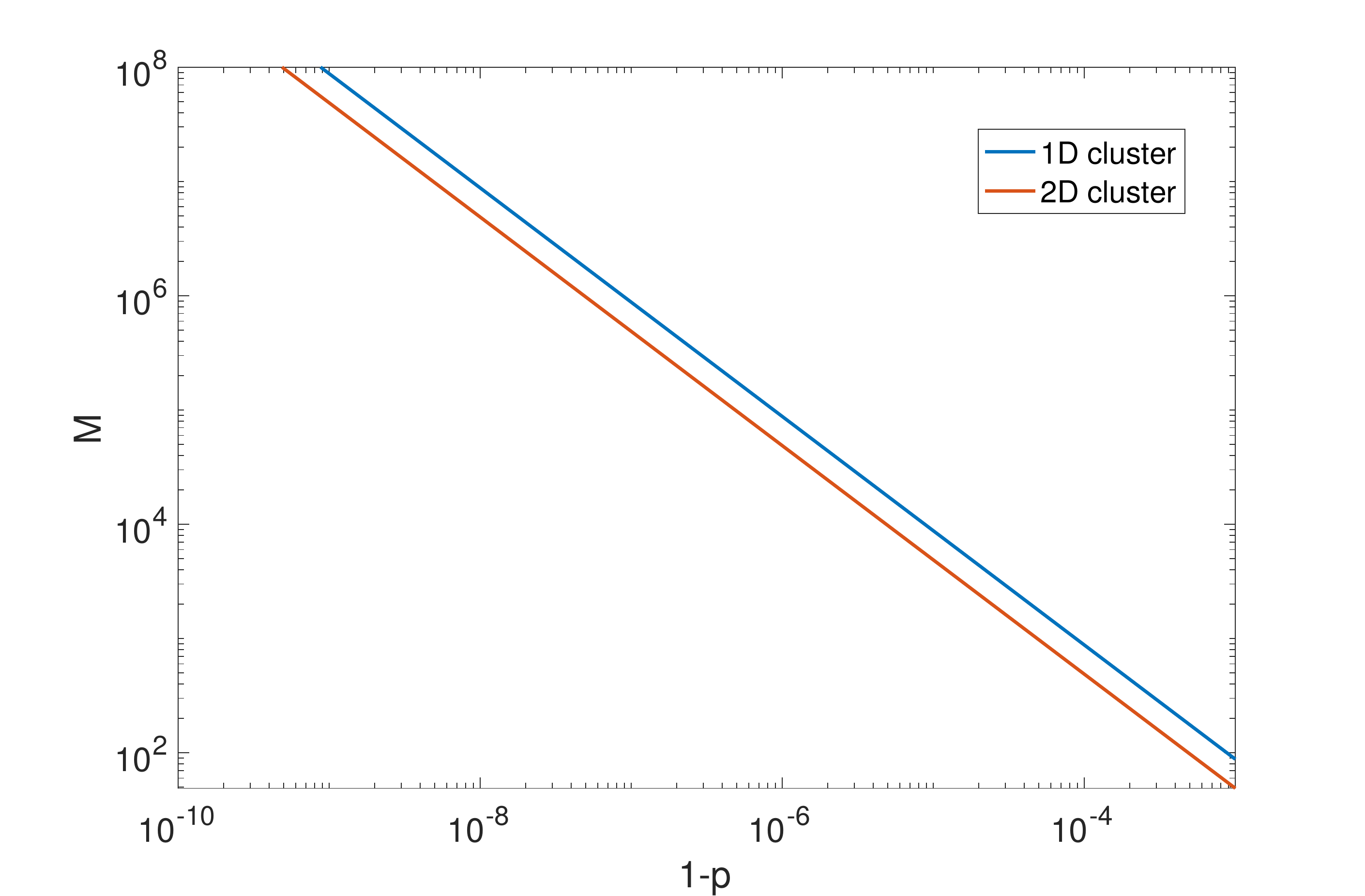} \label{CHSH1D2D}}
 \hfill
 \subfloat[\centering]{\includegraphics[width=1\columnwidth]{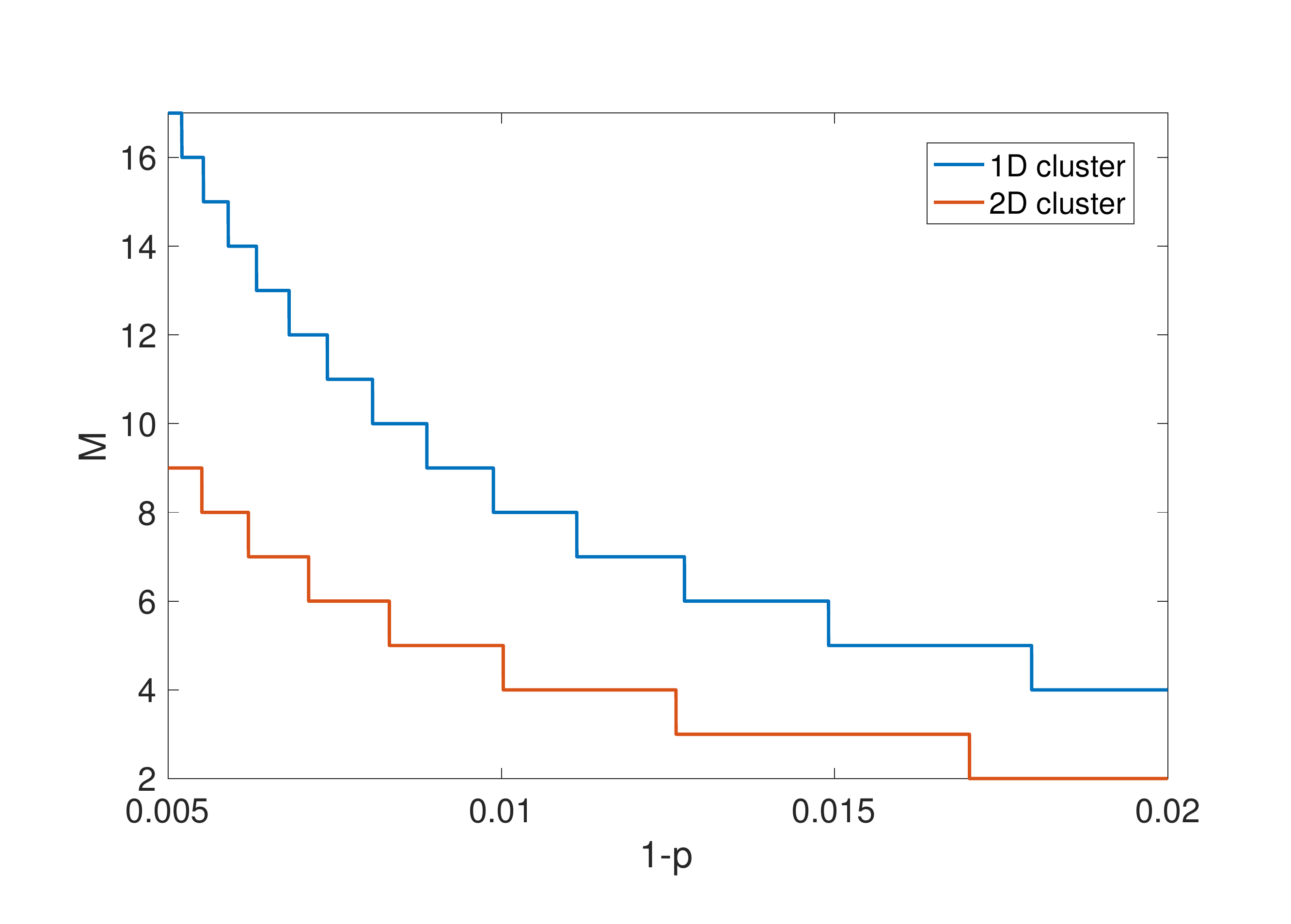} \label{SmallM}}
 \vspace{0cm}
\caption{(a) Plot of the maximal number of parties for which one can certify entanglement as a function of the noise $1-p$. (b) Similar plot for larger, experimentally better accessible values of $1-p$.}
\label{Figure_CHSH}
\end{figure}

The plots suggest a polynomial relation between $p$ and $M$.

In addition, we investigate a setup where 1D cluster states are generated via imperfect gates. Imperfect operations are modeled by LDN acting before the perfect operations. We use two different parameters $p_1$ and $p_2$ characterizing noise the of single- and two-qubit operations (for more details see the appendix). The maximal number of parties for which genuine multipartite entanglement can be shown is plotted in Fig. \ref{Figure} as a function of $p_1$ and $p_2$.

\begin{figure}[ht]
\centering
\includegraphics[scale=0.3]{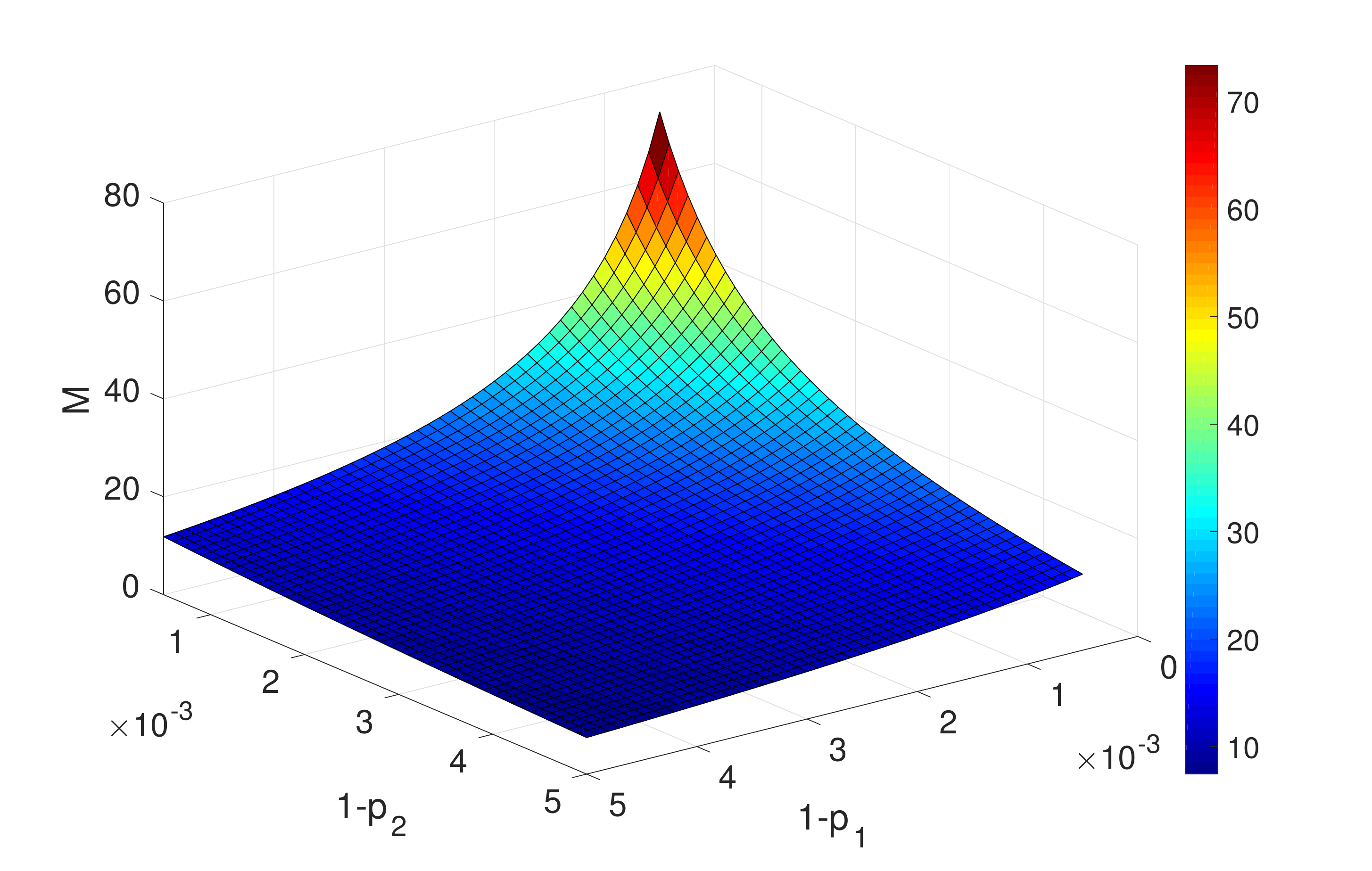}
\vspace{0mm}
\caption{Plot of the maximal number of parties for which one can certify genuine multipartite entanglement as a function of the noise $1-p_1$ and $1-p_2$. }
\label{Figure}
\end{figure}

\paragraph{Conclusion and outlook.---}
In this work we have introduced a scheme to detect genuine multipartite entanglement in a device-independent way based on bipartite Bell tests of entangled pairs that are deterministically generated from the initial state via LOCC. This allows not only for the detection of genuine entanglement of all pure states, but it is also applicable to mixed states with a sufficiently small amount of noise. The robustness of the scheme is directly related to the ratios of the local bound and the quantum violation of bipartite Bell inequalities, and any improvement on such inequalities directly leads to a larger set of states whose genuine entanglement can be certified device independently using our approach. While in general the scheme is not efficient, for important classes of states including the AKLT model and the 2D cluster state (which is a universal resource for measurement-based quantum computation), we have shown that our certification scheme is efficient and hence directly experimentally applicable.

We finally remark that a similar approach can be employed to reveal genuine non-locality for large classes of states \footnote{in preparation (2018)}, where in this case the criterion is more stringent as maximally entangled qubit pairs on a covering set need to be generated deterministically. Still, many multipartite states including the toric code, the ground state of the AKLT model or all connected graph states, can be shown to be genuine non-local.

\paragraph{Acknowledgments.---}
This work was supported by the Austrian Science Fund FWF (P28000-N27 and P30937-N27), the Swiss National Science Foundation SNSF (P300P2-167749 andÊ 200021-175527), the NCCR QSIT (PP00P2-150579), the Army Research Laboratory Center for Distributed Quantum Information via the project SciNet and the EU via the integrated project SIQS.

\section*{Appendix}
\label{appendix}
\section{Proof of Theorem 1}

We now construct the Bell expression. By the premise of the theorem, for each pair $E_k=(i_k, j_k)$ in the covering set $E$ there exists a LOCC protocol, such that each branch results in a pure entangled state  $\ket{ \Psi( {\bm \mu}^{(k)} )}$ shared by the parties $i_k$ and $j_k$, where ${\bm \mu}^{(k)}$ describes the branch of the LOCC protocol, i.e., it contains the ordered sequence of all the inputs (specifying the local operations) and outputs (labeling the obtained outcomes) of all the parties that lead to the state  $\ket{ \Psi( {\bm \mu}^{(k)} )}$. Note that we can always tailor the LOCC protocol such that all possible branches ${\bm \mu}^{(k)}$ occur (if a protocol eventually contains an operation with an outcome which never occurs for our input state one can simply merge this outcome with another one without altering the action of the protocol on that state). For any entangled bipartite state $\ket{ \Psi( {\bm \mu }^{(k)})}$ there exists a tailored Bell inequality which it is maximally violated by this state. In the two qubit states it is enough to consider the tilted-CHSH family of Bell inequalities \cite{tCHSH} for this purpose, and the general case has been recently shown in \cite{Scarani2017}. Let us denote the coefficients of this Bell test by $B_{{\bf y} | {\bf b}\, | {\bm \mu}^{(k)}}$, where $\bf y$ are the outputs and $\bf b$ the inputs of the measurements performed by the remaining parties $i_k$ and $j_k$ on the resulting state. In addition, we can always rescale the coefficients of the Bell expression such that the its maximal quantum value, attained by the state $\ket{ \Psi( {\bf x}^{(k)} {\bf a}^{(k)} )}$ for the correct measurement, is given by some predefined constant $\beta^*$, while the local bound $\beta_L<\beta^*$ remains strictly smaller. We now define the following Bell test
\be
B_{E_k} = \sum_{{\bm \mu }^{(k)} }\sum_{\bf y, b} B_{{\bf y} | {\bf b}\, | {\bm \mu}^{(k)}}
\ee
and consider its expected value for our initial state. The Bell test consists of ascribing the value $B_{{\bf y} | {\bf b}\, | {\bm \mu}^{(k)}}$ to the event where the LOCC protocol resulted in the branch $\bm \mu^{k}$ and the subsequent measurements of the parties $i_k$ and $j_k$ with setting $\bf b$ resulted in outcomes $\bf y$. The expected value of this Bell test on our state is hence given by value $\mn{B_k}= \beta^*$, since this is the expected value for each branch. 

Next consider the global Bell test which consists of the sum of $B_k$ defined above for all pairs $(i_k,j_k)$ appearing in the covering
\be
B = \sum_{E_k \in E} B_{E_k}.
\ee
Its expectation value on the initial state is simply given by $\mn{B} = \beta^* K$, the expected value $\beta^*$ of each $B_k$ times the number of pairs $K$ in the covering set $E$. This is still the maximal possible quantum value for the test, as it can not be improved for none of the pairs $E_k$ and none of the branches $\bm \mu^{(k)}$.

We now show that this value cannot be achieved by a bi-separable quantum state 
\be \label{app: bs}
\varrho_\textrm{\,BS}= \sum_\lambda\!\!\!\! \sum_{\substack{ \\ g_1\cap \,g_2 =\emptyset \\g_1\cup\, g_2 =\{1, \dots,  N \}}}\!\!\!\!
p\,(\lambda)\,  \rho_{g_1}(\lambda) \rho_{g_2}(\lambda)
\ee
Since $E$ corresponds to a connected graph for any term $g_1|g_2$ in the model \eqref{app: bs} there is at least one pair $E_k$  such that $i_k$ and $j_k$ are in different groups. By definition of LOCC, these parties $i_k$ and $j_k$ starting in different groups remain in a separable state in each branch of the protocol
\be
 \rho_{g_1}(\lambda) \rho_{g_2}(\lambda) \xrightarrow{\bm \mu^{(k)},\text{tr}_{\overline{i_k,j_k} }} \sum_\xi p(\xi) \rho_{i_k}(\xi) \rho_{j_k}(\xi). 
\ee
Hence the expectation value of the $B_{E_k}$ is bounded by some local bound $\beta_L<\beta^*$ which depends on the particular form of the inequality but is anyway strictly smaller than the quantum value attained by the target state. It follows for that the expected value of the global Bell test for a biseparable state
\be
\label{BBS}
\mn{B}_{BS} \leq (K-1)\beta^* + \beta_L < K \beta^*.
\ee
Hence, observing the value $K \beta^*$, or equivalently, reaching the maximal quantum value in each branch of the LOCC protocol for all pairs in the cover $E$ rules out a biseparable model. Therefore one can conclude that the statistics are only compatible with a non-separable model. This shows genuine multipartite entanglement in a device-independent way, since there are no assumptions on the underlying system except that it obeys the laws of quantum mechanics.

\section{Proof of Theorem 2}

Consider an entangled N-partite pure state $\ket{\Psi}$. Let $d$ be the dimension of the first-party Hilbert space supporting $
\ket{\Psi}$.  We will now show that theres exists a measurement basis for the first party, for which the post-measured state of the remaining $N-1$ parties is fully entangled for each outcome, i.e., all post-measurement states are entangled, or non-product, in each possible bipartition. By repeating the argument it is then possible to find measurements for all $N-2$ parties such that the final 2-partite state is entangled in all branches of possible measurement outcomes.

The Schmidt decomposition of the state $\ket{\Psi}$ in the splitting (first party)|(the remaining $N-1$ parties) reads
\be
\ket{\Psi} = \sum_{k=1}^d a_k \ket{k} \ket{R_k},
\ee
with $\braket{k}{\ell} = \braket{R_k}{R_\ell}=\delta_{k,\ell}$.
Let us choose the splitting of the remaining $N-1$ parties in two groups $G_1|G_2$. Importantly, there are only finitely many of such splittings, and for showing entanglement of the post-measurement state it is sufficient to show that the state is not product in any splitting.

Consider the projection of the first system onto the state 
\be
\ket{\xi }= \sum_{k=1}^d \xi_k \ket{k},
\ee
which yields the post-measurement state
\be
\ket{\Psi_\xi} \propto\braket{\xi}{\Psi} = \sum_k \xi_k^* a_k \ket{R_k}.
\ee

Note that for any such state the probability to get an outcome is nonzero $\sum_k | \xi_k a_k|^2 \neq 0$. 
We are interested in the set of states 
\be
\Xi = \{\ket{\xi}, \text{such that} \ket{\Psi_\xi} \propto \ket{\Omega_\xi}_{G_1} \ket{\Lambda_\xi}_{G_2} \}
\ee
 for which the post-measurement state $\ket{\Psi_\xi}$ is product in the $G_1|G_2$ splitting. In particular,
we will show that it is a set of measure zero within the $d$-dimensional Hilbert space of the first party $\cH_d$. 

To show this let us take a set of $d$ linearly independent states $\{\ket{\zeta_k}\in \Xi\}_{k=1}^d$ and denote the corresponding post-measurement states 
\be
\braket{\zeta_k}{\Psi} = c_k \ket{\Psi_k} = c_k \ket{\Omega_k}_{G_1}\ket{\Lambda_k}_{G_2} .
\ee
If such a set does not exist, then all the states in $\Xi$ belong to a $(d-1)$-dimensional subspace of $\cH_d$ at most and the proof is complete.

As the states $\ket{\zeta_k}$ are linearly independent they define a basis for $\cH_d$ (not necessarily orthonormal) , which can be used to uniquely express any state 
\be
\ket{\xi} = \sum_{k=1}^d z_k \ket{\zeta_k},
\ee
with the normalization constraint
\be
1= \braket{\xi}{\xi} = \sum_{k,\ell} \, z_\ell^*  \,G_{ \ell} z_k  = {\bf z}^\dag G \, {\bf z},
\ee 
where the gram matrix $G_{k \ell} = \braket{\zeta_\ell}{\zeta_k}$ is invertible due to the linear independence of $\{\ket{\zeta_k}\in \Xi\}_{k=1}^d$.
The post-measurement state reads
\be
\ket{\Psi_\xi} = \sum_k \underbrace{z_k c_k}_{\equiv y_k} \ket{\Omega_k}\ket{\Lambda_k}.
\ee
Hence, in terms of the new parameterization ${\bf y}$ we want to show that  among all the vectors ${\bf y}$ normalized as
\be\label{eq: normalisation}
{\bf y}^\dag \, C^{-1} G C^{-1}\,  {\bf y} =1
\ee
with $C = \text{diag}\{c_k\}_{k=1}^d $, the subset satisfying 
\be\label{eq: prod cond}
\ket{\Psi_{\bf y}} = \sum_k y_k \ket{\Omega_k}\ket{\Lambda_k} \propto \ket{\Omega_{\bf y}}\ket{\Lambda_{\bf y}}
\ee 
is of measure zero. In fact, the normalization of the state $\ket{\Psi_{\bf y}}$ as well as a global phase factor is irrelevant for the proportionality condition above, neither are they important for the question of the measure of the set of states satisfying Eq.~\eqref{eq: prod cond} (we can equivalently look at the cartesian product of normalized quantum states and nonzero complex number $(z, \ket{\Psi}) \simeq z \ket{\Psi}$). Hence, we have to show that the set of complex vectors ${\bf y}\in \mathbb{C}^{d}$ satisfying 
\be
\sum_k y_k \ket{\Omega_k}\ket{\Lambda_k} \propto \ket{\Omega_{\bf y}}\ket{\Lambda_{\bf y}}
\ee
is of measure zero within $\mathbb{C}^d$.

In addition we know that $\text{span}\{\ket{\Omega_k}\ket{\Lambda_k}\}_{k=1}^d =\text{span}\{\ket{R_k}\}_{k=1}^d$ is of dimension $d$, and 
\be
\dim\left(\text{span}\{\ket{\Omega_k}\}_{k=1}^d\right),\dim\left(\text{span}\{\ket{\Lambda_k}\}_{k=1}^d\right) \geq 2
\ee
as otherwise there is no possibility for the state $\ket{\Psi}$ to be entangled. Furthermore,
\be
\dim\left(\text{span}\{\ket{\Omega_k}\}_{k=1}^d\right) \times \dim\left(\text{span}\{\ket{\Lambda_k}\}_{k=1}^d\right) \geq d.
\ee

\subsection{Qubit case of $d=2$}
\label{sec: d=2}

In the case $d=2$, there are just two states $\ket{\Omega_1}\ket{\Lambda_1}$ and $\ket{\Omega_2}\ket{\Lambda_2 }$, which moreover are pairwise linearly independent $\braket{\Omega_1}{\Omega_2}\neq 1$ and $\braket{\Lambda_1}{\Lambda_2}\neq 1$. It follows that any state
\be
\ket{\Psi_{\bf y}}=y_1 \ket{\Omega_1}\ket{\Lambda_1} + y_2 \ket{\Omega_2}\ket{\Lambda_2} \notin \Xi\quad \text{for} \quad y_1 y_2 \neq 0 .
\ee
To see this we use the following observation:
\begin{proposition}
Given four states $\ket{a_1},\ket{a_2} \in \cH_\Omega$ and $\ket{b_1},\ket{b_2} \in \cH_\Lambda$, and a product state $\ket{\Psi}\in \cH_\Omega\otimes \cH_\Lambda$, the following always holds
\be
\braket{a_1,b_1}{\Psi} \braket{a_2, b_2}{\Psi} = \braket{a_1,b_2}{\Psi} \braket{a_2, b_1}{\Psi}.
\ee
\end{proposition}
\begin{proof}
For any product state $\ket{\Psi} = \ket{\Omega}\ket{\Lambda}$ one has
\begin{align}
\braket{a_1,b_1}{\Psi} \braket{a_2, b_2}{\Psi} &= \braket{a_1,b_1}{\Omega,\Lambda} \braket{a_2, b_2}{\Omega,\Lambda}\nonumber\\
\braket{a_1}{\Omega}\braket{b_1}{\Lambda} \braket{a_2}{\Omega}\braket{b_2}{\Lambda}&=
\braket{a_1}{\Omega}\braket{b_2}{\Lambda} \braket{a_2}{\Omega}\braket{b_1}{\Lambda}\nonumber\\
\braket{a_1,b_2}{\Omega,\Lambda} \braket{a_2, b_1}{\Omega,\Lambda} &= \braket{a_1,b_2}{\Psi} \braket{a_2, b_1}{\Psi}
\end{align}
which proves the proposition.
\end{proof}

Decomposing $\ket{\Omega_2} = \alpha \ket{\Omega_1} + \beta \ket{\Omega_1^\perp}$ and $\ket{\Lambda_2} = \gamma \ket{\Lambda_1} + \delta \ket{\Lambda_1^\perp}$, and using this observation for the states $\ket{a_1} = \ket{\Omega_1}, \ket{a_2} = \ket{\Omega_1^\perp}$, $\ket{b_1} = \ket{\Lambda_1}, \ket{b_2} = \ket{\Lambda_1^\perp}$ we get that a product state $\ket{\Psi_{\bf y}}$ must satisfy
\begin{align}
(y_1 + y_2 \alpha \gamma) (y_2 \beta \delta) = (y_2 \alpha \delta) (y_2 \beta \gamma) \Leftrightarrow\nonumber\\
y_1 y_2 \beta \delta =0 \Leftrightarrow y_1 y_2 =0, \label{eq: d=2},
\end{align}
where $\beta$ and $\delta$ are nonzero by $\braket{\Omega_1}{\Omega_2}\neq 1$ and $\braket{\Lambda_1}{\Lambda_2}\neq 1$. Hence, there are only the two states with either $y_1=0$ or $ y_2 =0$ that lead to a product post-measurement state.

\subsection{General case $d>2$}
We will proceed by recursion. 

First, let us rearrange the components $\{ \ket{\Omega_k,\Lambda_k}\}_{k=1}^d$ in such a way that  
\be\label{eq: LI}
\braket{\Omega_1}{\Omega_2}\neq 1 \quad \text{and}  \quad \braket{\Lambda_1}{\Lambda_2}\neq 1.
\ee 
From linear independence we get that $\braket{\Omega_1,\Lambda_1}{\Omega_2,\Lambda_2}\neq 1$, hence either $\braket{\Omega_1}{\Omega_2} \neq 1$, or $\braket{\Lambda_1}{\Lambda_2} \neq 1$ or both. If both inequalities hold there is nothing to rearrange. Otherwise, say $\braket{\Omega_1}{\Omega_2} \neq 1$ but $\ket{\Lambda_1}=\ket{\Lambda_2}$, consider the next components $\ket{\Omega_k,\Lambda_k}$ for $k\geq 2$. As $\dim( \text{span}\{\ket{\Lambda_k}\} \geq 2$ there has to be at least one state $\ket{\Omega_\ell,\Lambda_\ell}$ with $\braket{\Lambda_1}{\Lambda_\ell}\neq 1$. But in addition, either $\braket{\Omega_1}{\Omega_\ell}$ or  $\braket{\Omega_2}{\Omega_\ell}$ has to be not equal to one as $\ket{\Omega_1}$ and $\ket{\Omega_2}$ are different states. Hence, either the pair $(\ket{\Omega_1,\Lambda_1}, \ket{\Omega_\ell, \Lambda_\ell})$ or $(\ket{\Omega_2,\Lambda_2}, \ket{\Omega_\ell, \Lambda_\ell})$ satisfy \eqref{eq: LI}.
Next we show the following.

\begin{proposition}\label{prop: 2}
Given $d\geq 3$ linearly independent states $\{\ket{\Omega_k,\Lambda_k}\}_{k=1}^{d-1}$  such that the condition
\be
\sum_{k=1}^{d-1} x_k \ket{\Omega_k,\Lambda_k} \propto \ket{\Omega_{\bf x}, \Lambda_{\bf x}}
\ee
(for some $\ket{\Omega_{\bf x}, \Lambda_{\bf x}}$) is fulfilled by vectors $\bf x$ that form a set of measure zero inside $\mathbb{C}^{d-1}$, then for linearly independent states  $\{\ket{\Omega_k,\Lambda_k}\}_{k=1}^{d-1}\cup\{\ket{\Omega_d,\Lambda_d}\}$ the condition 
\be
\sum_{k=1}^{d} y_k \ket{\Omega_k,\Lambda_k} \propto \ket{\Omega_{\bf y}, \Lambda_{\bf y}}
\ee
(for some $\ket{\Omega_{\bf y}, \Lambda_{\bf y}}$) is fulfilled by vectors $\bf y$ forming a set of measure zero inside $\mathbb{C}^d$.
\end{proposition}
\begin{proof}
To show this, first consider  the sum of the first $(d-1)$ terms. It can be written in the Schmidt form 
\be
\sum_{k=1}^{d-1} y_k \ket{\Omega_k,\Lambda_k}= \sum_{\ell =1}^r b_\ell \ket{O_\ell, L_\ell},
\ee
with nonzero coefficients $b_\ell$ and some integer $r\geq 1$. We know that $r=1$ is anyway only possible for a measure zero set of $\bf y$ (from the $d-1$ case), hence we can ignore this case. 

For $r\geq 2$, from the Prop.~1 with $\ket{a_i}=\ket{O_i}$ and $\ket{b_i}=\ket{L_i}$ to be product the superposition $\ket{\Psi_{\bf y}} \propto \sum_{\ell =1}^r b_\ell \ket{O_\ell,L_\ell} + y_d \ket{\Omega_d,\Lambda_d}$ must satisfy
\begin{align}
&(b_i + y_d \braket{O_i}{\Omega_d}\braket{L_i}{\Lambda_d})(b_j + y_d \braket{O_j}{\Omega_d}\braket{L_j}{\Lambda_d }) =\nonumber\\
&(y_d)^2 \braket{O_i}{\Omega_d}\braket{L_i}{\Lambda_d} \braket{O_j}{\Omega_d}\braket{L_j}{\Lambda_d } \Leftrightarrow \\
&y_d = \frac{b_j b_j}{b_i \braket{O_i}{\Omega_d}\braket{L_i}{\Lambda_d} + b_j \braket{O_j}{\Omega_d}\braket{L_j}{\Lambda_d}},
\label{eq: y cons}\end{align}
for each pair of $i$ and $j$. This condition yields an equality constraint on the last element $y_d$. Note also that the constraint is unchanged if  ${\bf y}$ is multiplies by a  constant  ${\bf y} = \alpha {\bf x}$ (as for the normalization or the global phase factor) 
\begin{align}
\alpha \, x_d = \frac{\alpha^2 b_j b_j}{\alpha \, b_i \braket{O_i}{\Omega_d}\braket{L_i}{\Lambda_d} + \alpha \, b_j \braket{O_j}{\Omega_d}\braket{L_j}{\Lambda_d}}\nonumber\\
\Leftrightarrow  x_d = \frac{b_j b_j}{ b_i \braket{O_i}{\Omega_d}\braket{L_i}{\Lambda_d} +   b_j \braket{O_j}{\Omega_d}\braket{L_j}{\Lambda_d}}.
\end{align}
The set of vectors $\bf y$ fulfilling the desired condition \eqref{eq: y cons}, that we just derived, is of measure zero, since its last component $y_d$ is required to have a fixed value.

\end{proof}

Given our linearly independent states $\{ \ket{\Omega_k,\Lambda_k}\}_{k=1}^d$ fulfilling Eq.~\eqref{eq: LI}, we have that 
\begin{itemize}
\item From the result of the $d=2$ case the set of ${\bf y} \in \Xi_2 \subset \mathbb{C}^2$ for which
\be
\sum_{k=1}^2 y_k \ket{\Omega_k,\Lambda_k}
\ee 
is proportional to a product state is of measure zero.

\item By using the Proposition~\ref{prop: 2}  recursively we get that all sets $\Xi_d \subset \mathbb{C}^d$ of vectors $\bf y$ for which 
\be
\sum_{k=1}^d y_k \ket{\Omega_k,\Lambda_k}
\ee
is proportional to a product state is also of measure zero. This is precisely what we aimed to show.
\end{itemize}

\subsection{Implication for von Neumann measurements}

Consider an orthonormal basis $\{\ket{\psi_k}\}_{k=1}^d$ of the Hilbert space $\cH_d$, and let us denote the set of all bases $\mathbb{B}$. Now consider the subset $\mathbb{B}_{E}\subset \mathbb{B}$ which corresponds to von Neumann measurements for which the post-measurement $N-1$ partite states $\braket{\psi_k}{\Psi}$ are entangled in each branch $k=1\dots d$. This subset consists of all bases for which all states $\ket{\psi_d}$ are not inside $\Xi_d$, i.e., it is given by
\be
\mathbb{B}_{E} = \mathbb{B}\setminus \left(\bigcup_{k=1}^d \Xi_d[\ket{\psi_k}]\right).
\ee
As $\Xi$ is of measure zero on the set of states, $\Xi_d[\ket{\psi_k}]$ is of measure zero on the set of bases for each $k$. Hence, a finite union of such sets $\left(\bigcup_{k=1}^d \Xi_d[\ket{\psi_k}]\right)$ is also of measure zero, i.e., $\mathbb{B}_{E}$ is of full measure.

So far we considered the entanglement of the post-measurement state for a particular splitting $G_1|G_2$ of the remaining $N-1$ parties, i.e., $\mathbb{B}_{E}$  consists of all measurements for which the post-measurement states are entangled in the $G_1|G_2$ splitting. But we need the state to be genuinely entangled, in other words entangled with respect to all possible splittings $G_1|G_2$. The set of bases which satisfies this can be constructed from $\mathbb{B}$ by subtracting all bases for which at least one post-measurement state is product in at least one splitting, formally
\be
\mathbb{B}_{GME}= \mathbb{B}\setminus \left(\bigcup_{\text{all} \,\, G_1|G_2}\big( \bigcup_{k=1}^d \Xi_d^{G_1|G_2}[\ket{\psi_k}]\big)\right).
\ee
Again we only subtracted a finite union of sets of measure zero from $\mathbb{B}$, therefore $\mathbb{B}_{GME}$ is a set of full measure. So not only there exists a measurement of the first party for which all post-measurement states are genuinely entangled, but almost all measurements (except a measure zero subset) are like that.

\subsection{From $N$ to $2$ parties}

We have shown that almost any measurement brings the genuinely entangled $N$-partite state $\ket{\Psi}$ to $d_1$ post-measurement $(N-1)$-partite genuinely entangled states $\ket{\Psi_k}$. This procedure can be continued by measuring the next subsystem of dimension $d_2$, again almost any basis (in fact we do not even need to make it depend on $k$) yields $d_1 d_2$ genuinely entangled $(N-2)$-partite states. Finally, there exist local measurements for $
N-2$ parties yielding $d_1d_2\dots d_{N-2}$ genuinely entangled bipartite states on the remaining two parties. In fact, almost any choice of local measurements satisfies this, as again the construction of such a set corresponds to subtracting finitely many measure-zero states from the set of all possible local measurements. This means that there exist fixed local measurements which yield entangled bipartite states for any choice of the two parties.

\section{Generalized graph states}
Graph states were generalized to dimensions $d>2$ \cite{Zhou2003,Bahramgiri2006}. The generalized graph state of $N$ qudits can then be defined as
\be
\ket{G} = \prod_{\{i,j\} \in E} C_{ij}^{w_{ij}} \ket{+}^{\otimes N}.
\ee
Here, $C_{ij}$ is a controlled operation defined by $C_{ij}\ket{k}_i\ket{l}_j = \omega^{kl}\ket{k}_i\ket{l}_j$ with $\omega = e^{2\pi i/p}$. $\{\ket{k}\}_{k \in \{0, ..., p-1\}}$ denotes the computational basis of $\mathbb{C}^p$ and $\ket{+} = \tfrac{1}{\sqrt{p}} \sum_{k=0}^{p-1} \ket{k}$. Each edge $\{i,j\}$ has a weight $w_{ij} \in \mathbb{F}^p$, where $\mathbb{F}^p$ is a finite field of order $p$.

Similar to the qubit case one can cut out a pair of (connected) qudits by measuring the neighborhood in the computational basis \cite{Zhou2003,Bahramgiri2006}.

\section{More details on the AKLT model}
As discussed in the main text, one can probabilistically cut the chain by measuring a site in the basis $\{\ket{\tilde0}, \ket{\tilde1}, \ket{\tilde2}\}$, where one identifies the physical three-level system with the virtual correlation space via $\ket{\tilde0} = \ket{00}, \ket{\tilde1} = \ket{11}, \ket{\tilde2} = \ket{\psi^+}$. For the cases where one projects on $\ket{\tilde0}$ and $\ket{\tilde1}$ the chain gets decoupled, whereas for the $\ket{\tilde2}$ case one performs entanglement swapping on the virtual $\ket{\psi^-}$ states shared between neighboring sites. In this case one creates a virtual pair in the state $\ket{\psi^+}$ between the sites $i-1$ and $i+1$ and the (measured) site $i$ is removed from the chain. Thus one ends up in a situation similar to the original one, up to a Pauli $Z$ correction in the virtual space. One proceeds by measuring site $i+1$ in the basis $\{\ket{\tilde0}, \ket{\tilde1}, \ket{\tilde2}\}$. Again, if one obtains $\ket{\tilde0}$ and $\ket{\tilde1}$ the chain gets decoupled, otherwise one creates a $\ket{\psi^-}$ state shared between sites $i-1$ and $i+2$. This procedure is iterated until the chain is finally decoupled. In each step, the probability to decouple is given by two third. Hence, measuring a region of $n$ sites adjacent to each side of the pair which one would like decouple, leads to a success probability of $p_{\rm{cut}} = 1 - 2\left(\tfrac{1}{3}\right)^n$, which approaches unity exponentially fast in $n$.

There are eight different possibilities for the final state of the pair on which the Bell inequality will be tested. These are $P_{\rm{triplet}, A}P_{\rm{triplet}, B} \ket{i}_A\left(\ket{0}_A\ket{1}_B - (-1)^{\#}\ket{1}_A\ket{0}_B \right) \ket{j}_B$ up to normalization with $i, j \in \{0,1\}$ and $\#$ referring to the number of cases where one projected on $\ket{\tilde2}$ .  $P_{\rm{triplet}, A}$ denotes the projector on the triplet space on site $A$ and similarly for $B$. In each case the state is pure and entangled and thus there exists a Bell inequality for which it achieves the quantum bound \cite{tCHSH}.

The total number of branches which enter the global Bell inequality is $3^{2n}$. However, since the probability for cutting the chain approaches one exponentially fast, it is sufficient to choose $n$ logarithmically small in the system size $N$. Choosing $n=2\rm{log}_3 N$ leads to a success probability of $p_{\rm{suc}}=1-\tfrac{2}{N^2}$. We now assume that in the successful cases we obtain the quantum value $\beta_*$ and in the other cases the worst case scenario is to obtain $-\beta_*$. Then the observed value $\beta$ will be $\beta = \beta_* (1-\tfrac{4}{N^2})$. Using the results from the next section and assuming the worst case value of $\beta_*$ and $\beta_L$ (worst in the sense discussed below), one finds that genuine multipartite entanglement can be shown for up to 
\be
\label{M}
M = \tfrac{1}{4} N^2 (1- \tfrac{\beta_L}{\beta_*}).
\ee
Hence genuine multipartite entanglement among all parties can be shown in the large $N$ limit. This is because the right hand side of eq. (\ref{M}) can be made larger than the number of parties $N$ for sufficiently large $N$ (since $\beta_L < \beta_*$). For fixed, small $N$ one has to choose $n$ suitably.
The total number of branches is given by $3^{2n} = 3^{2\rm{log}_3 N} = N^2$, and thus the protocol is efficient in the system size. There is some freedom in the choice of $n$, e.g. choosing $n=(1+\epsilon)\rm{log}_3 N$ with $\epsilon > 0$ is also possible and leads to a total number of branches for each pair of $N^{1+\epsilon}$.

\section{More details on Dicke states}
We start with the $N$-qubit Dicke state with $k$ ($1 \leq k \leq N-1$) excitations,
\be
\ket{D_k^N} = {N \choose k}^{-1/2} \sum_{\mathrm{permutations}} \ket{1}^{\otimes k} \ket{0}^{\otimes N-k},
\ee
and observe that
\be
\ket{D_k^N} = \frac{1}{\sqrt{2}} \left( \ket{0} \ket{D_k^{N-1}} + \ket{1} \ket{D_{k-1}^{N-1}} \right).
\ee
This means that measuring one party in the computational basis results in a $N-1$ qubit Dicke state, $\ket{D_k^{N-1}}$ for the plus outcome and $ \ket{D_{k-1}^{N-1}}$ for the minus outcome, for the remaining $N-1$ parties.

Consider now a Dicke state with a single excitation. A measurement of the Pauli $X$ operator on one party will lead to the state (up to normalization) $\ket{D_1^{N-1}} \pm \ket{0}^{\otimes N-1}$, depending on the measurement outcome. A similar result is obtained for the case of a $N$-qubit Dicke state with $N-1$ excitations.

Now the LOCC protocol for producing an entangled state for any pair of parties, starting from a $N$-qubit Dicke with $k$ excitations is the following. One performs $Z$ measurements until one has a $M$-qubit Dicke state with either one or $M-1$ excitations, and then switches to $X$ measurements. One then obtains a bipartite state of the form $\chi = a(\ket{01} + \ket{10}) + \sqrt{1-2a^2}\ket{00}$ (or $\chi = a (\ket{01} + \ket{10}) + \sqrt{1-2a^2}\ket{11}$) with $a \neq0$, which is entangled.

Note that there are exponentially (in the system size) many branches in this LOCC protocol, as it involves measurements on $N-2$ parties. Hence it is not efficient, in contrast to the other examples discussed in this work.

\section{Impact of non-maximal Bell violations}
\label{impact}
Here we discuss the impact of non-maximal violations of the bipartite Bell inequalities, which could occur due to noise. We assume that the observed value of the Bell inequality for each pair $\beta$ satisfies $\beta_L < \beta < \beta_*$. This can be achieved in a general setting by rescaling all bipartite Bell inequalities such that they all have the same value $\beta*$. The largest value of all $\beta_L$ also provides a local bound for all other bipartite Bell inequalities, and similarly the smallest value of all $\beta$ can be used. Hence the results below can be applied in general.

\subsection{Approach 1}
For a single pair define $p$ as the probability that the two parties belong to the same group. Then the maximal observed value $\beta$ compatible with this assumption is $\beta \leq p \beta_* + (1-p) \beta_L $, hence we get a bound on the probability
\be
p \geq \frac{\beta-\beta_L}{\beta_*-\beta_L}.
\ee 
Let us now look at a multipartite scenario and assume that $n$ parties appear in the same group with probability $p_n$, as shown in Fig.~\ref{fig:cover}b. And the last pair (with the n-th party and a new one) belongs to  the same group with probability $p$. One has that the last parties can be together in two disjoint ways: either all the $n+1$ parties are together (probability $p_{n+1}$), or the last two are together but the first $n$ are not together (probability $p'$), hence
\be
p_{n+1} + p'= p.
\ee
However, we also have that $p'\leq 1-p_{n}$, implying 
\be
p_{n+1} = p - p' \geq p_n -(1-p).
\ee
From this we easily obtain
\be
p_{n+1} \geq 1- n (1-p).
\ee
As long as the probability for the $p_{n+1}$ is positive the observed correlations necessarily contain a fully non-separable term. Hence, a violation $\beta=\beta_* f$ for each pair in a chain/cover (where $f$ is some kind of fidelity) is sufficient to prove genuine $M$-entanglement for up to 
\be
\label{M1}
M=  \left \lfloor \frac{1}{1-p} \right \rfloor = \left \lfloor \frac{\beta_*-\beta_L}{\beta_*-\beta} \right \rfloor.
\ee

For a state where one can create a Bell state for each pair in the cover (e.g. graph states), one can use the CHSH inequality with $\beta_L=2$ and $\beta_*=2\sqrt{2}$. The equation above then becomes
\be
M =  \left \lfloor \frac{2-\sqrt{2}}{2}\frac{1}{1-f} \right \rfloor.
\ee

\begin{figure}[!ht]
\centering
\vspace{-3cm}
\includegraphics[width=0.98\columnwidth]{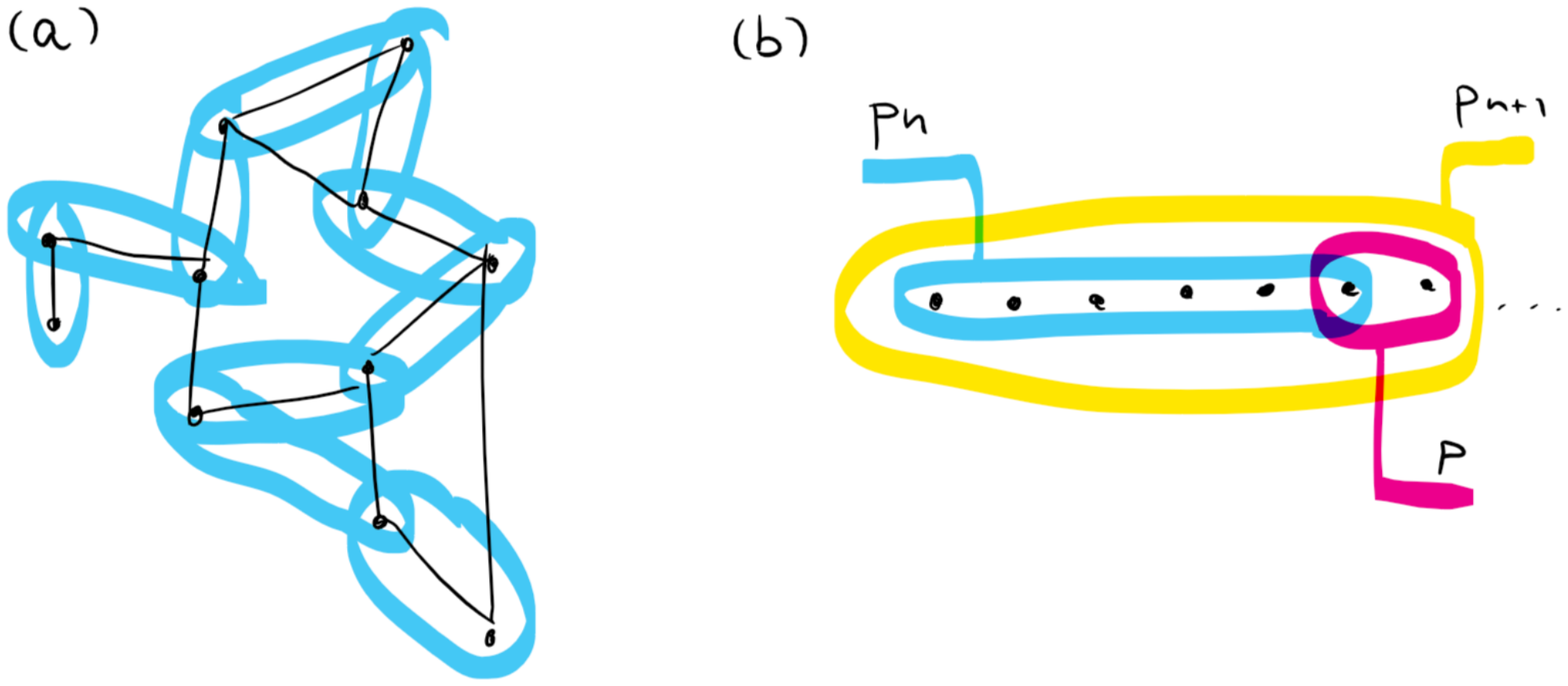}
\vspace{-4cm}
\caption{(a) Illustration of a covering set (blue ellipses). (b) Illustration for the section on the impact of non-maximal Bell violations.}
\label{fig:cover}
\end{figure}

\subsection{Approach 2}

Alternatively one can derive the impact of non-maximal violations of the bipartite Bell inequalities from eq. (\ref{BBS}). We assume that there are $N-1$ pairs in the covering set $E$, and that one obtains a value of $\beta$ for the Bell operator for each of them. A biseparable model can at most reach a value of $(N-2)\beta_* + \beta_L$, whereas the observed quantum value is $(N-1)\beta$. This shows that for sufficiently large $\beta$ the observed value is incompatible with a biseparable model, and hence shows genuine $N$-partite entanglement. A straightforward calculation gives
\be
N = \left \lfloor \frac{\beta_*-\beta_L}{\beta_*-\beta} + 1 \right \rfloor,
\ee
which is up the term $+1$ identical to eq. (\ref{M1}).

When $\beta$ is too small to show the incompatibility with a biseparable model, one might still be able to show the incompatibility with an $n$-separable model. Such a model can at most reach a value of $(N-n)\beta_* + (n-1)\beta_L$, which for sufficiently large $\beta$ cannot explain the observed value $(N-1)\beta$. For the number of genuinely entangled parties $M=\tfrac{N}{n}$ one then obtains
\be
M =  \left \lfloor \frac{N(\beta_*-\beta_L)}{N(\beta_*-\beta) + \beta - \beta_L} \right \rfloor,
\ee
which reproduces eq. (\ref{M1}) in the limit $N \rightarrow \infty$.

\section{Preparation of cluster states with noisy gates}
We simulate the generation of cluster states using noisy operations in the following way. The imperfect initialization of qubits is modeled by perfect initialization followed by LDN parametrized by $p_1$. Noisy single- and two-qubit gates are modeled by LDN parametrized by $p_1$ and $p_2$, followed by the ideal gates. Finally an imperfect single-qubit measurement is described by LDN parametrized by $p_1$ followed by the perfect measurement.

\bibliographystyle{apsrev4-1}
\bibliography{DIbib}

\end{document}